\documentclass[11pt,oneside]{article}
\usepackage{import}
\usepackage{etoolbox}

\usepackage[font=small,skip=0pt]{caption}

\providebool{DRAFT}
\booltrue{DRAFT}

\usepackage{etoolbox}
\usepackage{xparse}

\makeatletter%
\@ifclassloaded{book}
{%
  \newcommand{\whenbook}[1]{%
    #1%
  }%
}
\makeatother%

\providecommand{\whenbook}[1]{}%

\providebool{bookdraft}%

\NewDocumentEnvironment{When}{m +b}{%
  \providebool{#1}%
  \ifbool{#1}{#2}{}%
}{}%

\NewDocumentEnvironment{Unless}{m +b}{%
  \providebool{#1}%
  \ifbool{#1}{}{#2}%
}{}%

\providebool{DRAFT}
\ifbool{DRAFT}{
  \newcommand{\whendraft}[1]{#1}%
    \newcommand{\unlessdraft}[1]{}%
  }

\NewDocumentEnvironment{Draft}{+b}{%
  \whendraft{#1}%
}{%
}%

\NewDocumentEnvironment{DRAFT}{O{BlueViolet}}{%
  \begin{Draft}%
    \color{#1}%
  }{%
  \end{Draft}
}%

\usepackage[hyphens]{url}

\usepackage[style=alphabetic,natbib=true,maxalphanames=4,minalphanames=3,maxnames=10,doi=true,isbn=false,url=false,backend=biber]{biblatex} %

\AtEveryBibitem{%
\ifentrytype{proceedings}{
    \clearfield{editor}%
    \clearname{editor}%
  }{}
  \ifentrytype{inproceedings}{
      \clearfield{editor}%
      \clearname{editor}%
}{}
}

\usepackage{amsmath}                               %
\usepackage{amsthm}                                %
\usepackage[anchorcolor=blue,colorlinks]{hyperref}%
\usepackage{varioref}
\usepackage[capitalize]{cleveref}

\crefname{invariantsi}{invariant}{invariants}

\usepackage[margin=1in]{geometry}

\usepackage{fancyhdr}

\usepackage[spacing,kerning=true]{microtype}          %
\usepackage[T1]{fontenc}

\usepackage{article} %
\usepackage{math} %
\usepackage{algo} %
\usepackage{wrapfig} %
\usepackage{placeins} %
\usepackage{advdate}%

\usepackage{appendix}

\usepackage{graphicx} %
\usepackage{layout} %
\usepackage{setspace} %
\usepackage{mathtools} %
\usepackage{grffile} %
\usepackage{pdfpages} %
\usepackage{mdframed} %
\usepackage{bm} %
\usepackage{needspace} %
\usepackage{pbox} %
\usepackage{cancel}
\usepackage[normalem]{ulem}

\usepackage{tikz}        
\usetikzlibrary{arrows.meta,calc}

\usepackage[full]{textcomp}     %
\usepackage{lmodern}            %
\usepackage[T1]{fontenc}        %
\usepackage{inconsolata}        %

\usepackage{titlesec}

\usepackage{truncate}

\NewDocumentCommand{\undernote}{s O{blue} m m}{
  \IfBooleanT{#1}{\smash}%
  {\color{#2} %
    \underbrace{\normalcolor%
      #4}_{\mathclap{\text{#3}}} %
  }%
  \IfBooleanT{#1}{\vphantom{#4}}
}

\newcommand{\defterm}[1]{{\boldmath\normalfont \bfseries #1}}%
\renewcommand{\defterm}{\emph}%

\makeatletter
\g@addto@macro\bfseries{\boldmath}
\makeatother

\titlespacing*{\paragraph}{%
  0pt}{
  {\medskipamount}}{
  1em}

\titleformat{\subparagraph}[runin]{\itshape}{0pt}{}{}%

\titlespacing*{\subparagraph}{%
  0pt}{
  {\medskipamount}}{
  1em}

\newlength{\myalphabet}                %
\newlength{\mywidth}                   %
\newlength{\mymargin}                  %

\usepackage{skak}%

\bibliography{references}

\providecommand{\poly}{\fparnew{\operatorname{poly}}}

\definecolor{calypso}{RGB}{50, 104, 145} %

\hypersetup{allcolors=calypso}  %

\definecolor{almostblack}{RGB}{18, 18, 18} %

\color{almostblack}%
\hypersetup{allcolors=almostblack}  %

\providebool{article}
\booltrue{article}

\begin{document}

\providecommand{\incut}{\fparnew{\partial^{-\!}}}
\newcommand{\optS}{S^{\star}}   %
\newcommand{\MFTime}{\fparnew{\mathcal{T}}}%
\providecommand{\capacity}{\fparnew{c}}        %
\newcommand{\flow}{\fparnew{f}}%
\newcommand{\optT}{T^{\star}}  %
\newcommand{\comp}[1]{\vphantom{#1}\bar{#1}}%
\newcommand{\optt}{t^{\star}}               %
\newcommand{\involume}{\fparnew{\operatorname{vol}^-}} %
\newcommand{\supersink}{t^{\star}}                     %
\providecommand{\incutcap}[1]{\capacity{\incut{#1}}}         %
\providecommand{\involume}{\fparnew{\operatorname{vol}^-}} %
\providecommand{\incutsize}{\incutcap}                     %
\providecommand{\cutsize}{\capacity}                       %
\providecommand{\optlambda}{\lambda^{\star}}               %
\providecommand{\indegree}{\fparnew{\operatorname{deg}^-}} %
\providecommand{\outdegree}{\fparnew{\operatorname{deg}^+}} %
\providecommand{\vincut}{\fparnew{N^-}} %
\providecommand{\voutcut}{\fparnew{N^+}} %
\providecommand{\inneighbors}{\fparnew{N^-}}
\providecommand{\outneighbors}{\fparnew{N^+}}
\providecommand{\XXin}[1]{#1_{\textsc{in}}}%
\providecommand{\XXout}[1]{#1_{\textsc{out}}}%
\providecommand{\tin}{\XXin{t}}            %
\providecommand{\tout}{\XXou{t}}            %
\providecommand{\Tin}{\XXin{T}}             %
\providecommand{\Tout}{\XXout{T}}           %
\providecommand{\uout}{\XXout{u}}          %
\providecommand{\uin}{\XXin{u}}  %
\providecommand{\vout}{\XXout{v}}          %
\providecommand{\vin}{\XXin{v}}  %
\providecommand{\rout}{\XXout{r}} %
\providecommand{\xin}{\XXin{x}}   %
\providecommand{\xout}{\XXout{x}} %
\renewcommand{\xout}{\XXout{x}}   %
\providecommand{\Xin}{\XXin{X}}   %
\providecommand{\Xout}{\XXout{X}} %
\providecommand{\Uin}{\XXin{U}} %
\providecommand{\Uout}{\XXout{U}} %
\providecommand{\Tin}{\XXin{T}}   %
\providecommand{\Tout}{\XXin{T}}   %

\newcommand{\authornote}{%
  \texttt{krq@purdue.edu}.  Purdue University, West Lafayette,
  Indiana. Supported in part by NSF grant CCF-2129816.}

\title{Approximating Directed Connectivity in Almost-Linear Time}

\author{Kent Quanrud}

\maketitle

\begin{abstract}
  We present randomized algorithms that compute $\epsmore$-approximate
  minimum global edge and vertex cuts in weighted directed graphs in
  $\bigO{\log{n}^4 / \eps}$ and $\bigO{\log{n}^5/\eps}$
  single-commodity flows, respectively. With the almost-linear time
  flow algorithm of \cite{CKL+22}, this gives almost linear time
  approximation schemes for edge and vertex connectivity. By setting
  $\eps$ appropriately, this also gives faster exact algorithms for
  small vertex connectivity.

  At the heart of these algorithms is a divide-and-conquer technique
  called ``shrink-wrapping'' for a certain well-conditioned rooted
  Steiner connectivity problem. Loosely speaking, for a root $r$ and a
  set of terminals, shrink-wrapping uses flow to certify the
  connectivity from a root $r$ to some of the terminals, and for the
  remaining uncertified terminals, generates an $r$-cut where the sink
  component both (a) contains the sink component of the minimum
  $(r,t)$-cut for each uncertified terminal $t$ and (b) has size
  proportional to the number of uncertified terminals. This yields a
  divide-and-conquer scheme over the terminals where we can divide the
  set of terminals and compute their respective minimum $r$-cuts in
  smaller, contracted subgraphs.
\end{abstract}

\unmarkedfootnote{\authornote}

\section{Introduction}

Edge and vertex connectivity are two of the more basic problems in
graph algorithms.  For a directed graph $G = (V,E)$, the \defterm{edge
  connectivity} of $G$ is the minimum weight of any subset of edges
$F \subseteq E$ for which $G - F$ is not strongly connected.  The
minimizing set of edges $F$ is called the \defterm{minimum (global)
  edge cut}.  The \defterm{vertex connectivity} is the minimum weight
of any subset of vertices $U \subseteq V$ for which $G - U$ is not
strongly connected, and the minimizing set of vertices $U$ is called
the \defterm{minimum (global) vertex cut}.  Either problem can be
solved in polynomially many single-commodity flows. Here we are
interested in faster algorithms.

Let $m = \sizeof{E}$ denote the number of edges and $n = \sizeof{V}$
the number of vertices in $G$. For edge connectivity, the fastest
algorithms for (weighted) directed minimum cut run in strongly
polynomial $\bigO{m n \log{n^2 / m}}$ time \cite{HaoOrlin} and in
weakly polynomial $\min{m^{1+o(1)} \sqrt{n}, m^{2/3+o(1)} n}$
randomized time \cite{CLNPQS21}. There is also an
$\bigO{m \lambda \log{n^2 / m}}$-time algorithm for uncapacitated
graphs with edge connectivity $\lambda$ \cite{gabow-95}, and an
$\apxO{n^2}$ randomized time algorithm for simple graphs (with no
parallel edges) \cite{cq-simple-connectivity}. ($\apxO{\cdots}$ hides
logarithmic factors.)  All these algorithms directly solve the closely
related \emph{rooted connectivity} problem. For a fixed root
$r \in V$, the rooted edge connectivity is the minimum weight of any
cut $F \subseteq E$ such that $r$ cannot reach all of $V$ in $G - F$;
the minimizing cut is called the minimum \defterm{$r$-cut}. Edge
connectivity can be solved by rooted edge connectivity by taking any
vertex as the root and considering both $G$ and the graph reversing
all the edges in $G$.

Directed edge connectivity algorithms might take inspiration from the
remarkable success of undirected edge connectivity.
For undirected minimum cut, there has
been a randomized, $\apxO{m/\eps^2}$-time $\epsmore$-approximation
algorithm since 1993 \cite{karger-93}, and a nearly linear time exact
algorithm since 1996 \cite{karger-00}. A recent line of work has also
derandomized these bounds \cite{KT19,LP20,Li21,HLRW24}.  These
achievements beg to ask if there are nearly linear time algorithm for
directed graphs.

For weighted directed vertex connectivity, until very recently, the
state of the art was $\apxO{mn}$ randomized time \cite{hrg-00}.  In
unweighted directed graphs the minimum vertex cut can be computed in
$\bigO{\lambda^{5/2} m + m n}$, $\bigO{\lambda m n^{3/4} + mn}$,
$\apxO{\lambda^2 m}$, $\apxO{n \lambda^3 + \lambda^{3/2} \sqrt{m} n}$,
$\apxO{\lambda n^2}$, $mn^{1 - 1/12 +o(1)}$, or $n^{2+ o(1)}$ time,
where $\lambda$ is the vertex connectivity
\cite{Gabow2000b,NSY19,FNSYY20,cq-simple-connectivity,LNPSY25}. A
$\epsmore$-minimum vertex cut can be computed in $\apxO{n^2 / \eps}$
time in simple graphs and $\apxO{n^2 / \eps^2}$ in weighted graphs
\cite{cq-simple-connectivity,CLNPQS21}.  Similar to edge connectivity,
undirected graphs have faster algorithms for vertex
connectivity. There is an almost-linear time algorithm in unweighted
undirected graphs \cite{LNPSY25}, and recently, \cite{CT25} gave
$mn^{0.99 + o(1)}$ and $m^{1.5 + o(1)}$ time algorithms in weighted
undirected graphs.  Very recently, \cite{JMY25} gave a reduction from
directed vertex connectivity to undirected vertex connectivity on a
dense graph with $\bigO{n}$ vertices and $\bigO{n^2}$ edges. Thus
\cite{CT25} translates to a $\bigO{n^{2.99+o(1)}}$-time algorithm for
weighted directed graphs.

The discussion above highlights only recent and state-of-the-art
bounds out of a much larger literature of connectivity algorithms.
For a more comprehensive survey of previous work, see
\cite[\S15.2a and \S15.3a]{Schrijver}.

Several of the algorithms (e.g.,
\cite{LP20,cq-simple-connectivity,CLNPQS21,LNPSY25,CT25}) mentioned
above use $(s,t)$-maximum flow as a black box.  Maximum flow has
recently seen tremendous theoretical improvements, culminating in an
almost-linear, $m^{1 + o(1)}$ running time for polynomially bounded
capacities \cite{CKLPGS22,BCK+23}.

\paragraph{Almost-linear time approximations.} This work puts forth
the first almost-linear time approximation algorithms for both edge
and vertex connectivity. Here we are given an error parameter
$\eps > 0$, and desire edge or vertex cuts with weight at most
$\epsmore$-factor greater than the true minimum weight.  We solve
either connectivity problem with $\bigO{\polylog{n} / \eps}$ single
commodity flows; the almost-linear running time then follows from
\cite{CKLPGS22}.

Formally, for directed edge connectivity, we prove the following.
\begin{theorem}
  A $\epsmore$-minimum rooted and global edge cut in a directed graph
  can be computed with high probability in randomized running time
  bounded by $\bigO{\log{n}^4/\eps}$ instances of single-commodity
  flow over graphs of $\bigO{m}$ edges and polynomially bounded
  capacities. By \cite{CKLPGS22}, this is at most $m^{1+o(1)} / \eps$
  randomized time.
\end{theorem}

For directed vertex connectivity, we prove the following.
\begin{theorem}
  \labeltheorem{apx-vc}
  A $\epsmore$-minimum global (respectively, rooted) vertex cut can be
  computed with high probability in running time bounded by
  $\bigO{\log{n}^5 / \eps}$ (respectively, $\bigO{\log{n}^4/\eps}$)
  instances of single-commodity flow over graphs of $\bigO{m}$ edges
  and polynomially bounded capacities. By \cite{CKLPGS22}, this is at
  most $m^{1+o(1)} / \eps$ randomized time.
\end{theorem}

These fast approximation algorithms are attractive in situations where
the graph is so large that one is willing to exchange polynomial
factors in the input size for a controlled amount of error. Based on
the history of developments for undirected minimum cut
\cite{Karger93,Karger00}, we also hope they are a stepping stone
towards nearly linear time exact algorithms for directed minimum cut.

\paragraph{Exact algorithms for small connectivities.}
As mentioned above, there is interest in faster exact algorithms in
integer-capacitated graphs with small connectivity. (These are the
algorithms with running times parameterized by the connectivity
$\lambda$.) By taking $\eps$ inversely proportional to the value of
the minimum vertex cut, the $\epsmore$-approximate vertex cut
algorithm of \reftheorem{apx-vc} becomes a state-of-the-art exact
algorithm for the small vertex connectivity setting. Additional minor
adjustments improve the logarithmic factors to the following.

\begin{theorem}
  \label{small-vc}
  In integer capacitated graphs, the minimum global (respectively,
  rooted) vertex cut can be computed with high probability in running
  time bounded by $\bigO{\lambda \log \lambda \log^4 n}$
  (respectively, $\bigO{\lambda \log \lambda \log^3 n}$) instances
  of single-commodity flow over graphs of $\bigO{m}$ edges and
  polynomially bounded capacities, where $\lambda$ is the global
  (respectively, rooted) vertex connectivity. By \cite{CKLPGS22}, this is at
  most $\lambda m^{1+o(1)}$ randomized time.
\end{theorem}

For edge connectivity, the same idea gives an exact algorithm for
integer capacitated graphs that is within a subpolynomial factor of
the $\bigO{\lambda m \log{n^2 / m}}$ running time of
\cite{Gabow1995a}. It is structurally interesting to learn that an
exact minimum directed cut can be computed with high probability in
$\bigO{\lambda \polylog{n}}$ flows over the input graph.

\section{High-level overview and techniques}

Consider the rooted edge connectivity problem, where we have a fixed root
$r$ and want to find the minimum $r$-cut.  Obviously, the minimum
$r$-cut can be computed in $(n-1)$ (single-commodity) max-flow min-cut
computations, taking the minimum $(r,t)$-cut for every choice of
$t \in V - r$. Now suppose the sink-side of the minimum $r$-cut has
$k$ vertices. Then one can compute the minimum $(r,t)$-cut for just
$\apxO{n / k}$ uniformly sampled sink vertices $t$ and obtain the
minimum $r$-cut with high probability. This works well when $k$ is
large and can be balanced against other strategies when $k$ is
small. For example, when edge weights are integral and the cut value
$\lambda$ is small, local flow algorithms give $\poly{k,\lambda}$
subroutines that can be applied to each sampled sink $t$
\cite{nsy-19,fnsyy}. Alternatively, in simple graphs,
\cite{cq-simple-connectivity} observed that vertices with in-degree at
least $2k$ can be contracted into the root while preserving the
minimum $r$-cut, leaving a sparse graph with $\bigO{n k}$
edges. Balancing the size of the sink component against the sparsity of
the contracted graph leads to a nearly $n^2$ running time in simple graphs.

A third approach for undirected minimum cut is based on the symmetric
submodularity of undirected cuts. Suppose $G$ is undirected and the
smaller side of the minimum cut has $k$ vertices. Let $T \subseteq V$
sample $n /k$ vertices uniformly at random; with constant probability,
$T$ has exactly one vertex on the smaller side of the minimum cut. In
particular, the minimum cut is now an \emph{isolating cut} for $T$, in
that it isolates some $t \in T$ from $T - t$. The minimum isolating
cut for $T$ can be solved in $\bigO{\log n}$ single-commodity flows
and cuts and the key observation is that for any partition $(A,B)$ of
$T$, the minimum $(A,B)$-cut does not cross any minimum isolating cut
\cite{LP20}.

For rooted connectivity, the isolating-cut approach is very attractive
as it only requires $\bigO{\log n}$ maximum flows, independently of
the size of the sink component. It is also easy to sample a set of
vertices $T$ with exactly one vertex from the sink component of the
minimum $r$-cut. But symmetry plays a critical role in the isolating
cuts argument: the $(A,B)$-cuts partition the graph while preserving
the isolating cuts on \emph{both} sides, leading to a
divide-and-conquer dynamic that ultimately partitions the graph, with
each part containing a distinct isolating cut. In directed graphs, the
minimum $(A,B)$-cut has no relation to the minimum $(B,A)$-cut, and it
is hard to see how to use this technique to decompose the graph like
so.

All that said, fix a cut value $\lambda > 0$ and a set of terminals
$T \subseteq V - r$.  Suppose we want to find some terminal $t \in T$
with $r$-connectivity less than $\lambda$, or otherwise certify that
all $t \in T$ have $r$-connectivity (at least) $\lambda$. To try to do
this in a single flow computation, one can create a flow network that
tries to route $\lambda$ units of flow from $r$ to every $t \in
T$. (Here one introduces a super-sink $\supersink$ and arcs
$(t,\supersink)$ of capacity $\lambda$ for all terminals $t \in T$.)

\begin{center}
  \includegraphics{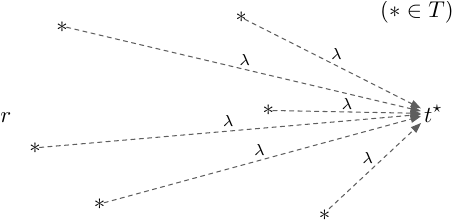}
\end{center}

The maximum flow in this network may saturate (i.e., route $\lambda$
units of flow to) some terminals $t \in T$, certifying the
$r$-connectivity to be at least $\lambda$ for each of these terminals.
What about the remaining, unsaturated terminals? Letting $(A,B+\supersink)$
denote the minimum cut in this flow network, these are the terminals
in $T \cap B$. \Reflemma{wrap} below proves the following uncrossing
claim: for each unsaturated terminal $t \in T \cap B$, there is a
minimum $(r,t)$-cut where the sink component is contained in $B$. The
proof ultimately comes down to submodularity of cuts although the
application is indirect.  \Reflemma{wrap} implies that one can take
the graph contracting $A$ into the root, $G/A$, and continue trying to
certify the $r$-connectivity of the remaining terminals $T \cap B$ in
$G/A$.

\begin{center}
  \includegraphics{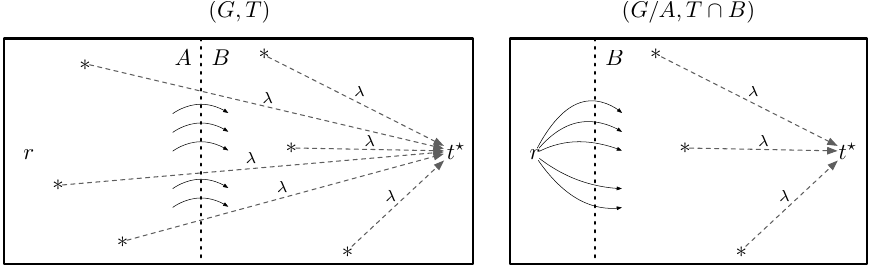}
\end{center}

Suppose we attempted the same kind of flow problem with $G / A$ and
$T\cap B$, trying to route $\lambda$ units of flow from $r$ to each
$t \in T\cap B$. Then, of course, one gets the same cut as before,
induced by $B$. One way to make progress is to halve $T \cap B$ into
two terminal sets $T_1$ and $T_2$, and consider the corresponding flow
problems for each subset of terminals separately. These smaller flow problems
generally produce different, refined cuts. We can continue to
divide-and-conquer over the terminals in this fashion until in the
base case we have one terminal $t$ that can be addressed with an
$(r,t)$-flow.

Alas, the running time of this divide-and-conquer algorithm is still
$m^{1+o(1)} \sizeof{T}$, no better than computing the minimum
$(r,t)$-cut directly for every $t \in T$. The problem is that the
contracted graphs $G / A$ may not be significantly smaller than
$G$. This deficiency motivates the second key idea of
\emph{preconditioning} the graph.

For each vertex $v \in V$, let $\indegree{v}$ denoted the unweighted
in-degree of $v$. For a set of vertices $X \subseteq V$, we define the
\defterm{in-volume} of $X$, denoted $\involume{X}$, as the sum of
unweighted in-degrees over all vertices in $X$,
$\involume{X} = \sum_{v \in X} \indegree{v}$.  Let $\mu$ be the
in-volume of the minimum $r$-cut.  For a parameter $\varphi$, we say
that $(G,r)$ is \defterm{rooted $\varphi$-conditioned} if every
$r$-cut $\incut{X}$, where $X \subseteq V - r$, has capacity at
least $\varphi \involume{S}$. For
$\varphi = \bigOmega{\epsilon \lambda / \mu}$, one can make $(G,r)$
$\varphi$-conditioned by adding, for every vertex $v \in V - r$, an
arc ($r,v$) of weight $\epsilon \lambda \indegree{v} / \mu$, while
increasing the size of the minimum $r$-cut by at most
$\epsilon \lambda$.

(A similar preconditioning step, adding arcs from $r$ to every other
vertex, appears in \cite{CLNPQS21}. The auxiliary arcs in
\cite{CLNPQS21} are uniformly weighted and are designed to facilitate
a randomized edge sparsification argument.)

Now, suppose $(G,r)$ is $\varphi$-conditioned for
$\varphi = \bigOmega{\epsilon \lambda / \mu}$, and let $T$ sample each
non-root vertex $v$ independently with probability
$\indegree{v} / \mu$. $T$ has $m / \mu$ vertices in expectation, and
at least one vertex from the sink component of the minimum $r$-cut
with constant probability. Consider the same flow network as before,
trying to route $\lambda$ units of flow to every $t \in T$. As before,
let $(A,B)$ be the minimum cut, and consider the unsaturated terminals
$T \cap B$ in the contracted graph $G / A$. The max-flow min-cut
theorem implies that the size of the cut, $\incutsize{B}$, is at most
the total sink-capacity of $B$, $\lambda \sizeof{T \cap B}$.  Since
$G$ is $\varphi$-conditioned,
$\incutsize{B} \geq \varphi \involume{B}$. But $\involume{B}$ is also
the number of edges in $G / A$! Putting everything together, the
number of edges in $G / A$ is at most
$\bigO{\mu \sizeof{T \cap B} / \epsilon}$.

The key point is that the number of edges in the contracted graph,
$G / A$, is proportional to the number of terminals $T \cap B$. Now
when we halve the unsaturated terminals $T \cap B$ into $T_1$ and
$T_2$, we also halve the number of edges in the graph. This dynamic
leads to a running time bounded by logarithmically many max-flows.

We call this the ``shrink-wrap'' algorithm. The auxiliary flow problem
produces a cut $(A,B)$ where $B$ ``wraps'' the minimum $(r,t)$-cut for
every unsaturated terminal $t$. Because $G$ is preconditioned,
contracting $A$ ``shrinks'' the graph, keeping the input size
proportional to the number of remaining terminals. Combined with
standard techniques, like guessing $\mu$ up to a constant factor in
parallel, and binary-searching for $\lambda$, we obtain an
almost-linear time approximation algorithm for minimum rooted cut.
Exactness was compromised in the preconditioning step when we added
the auxiliary arcs. The whole algorithm is arguably simple, in
hindsight.

The shrink-wrap algorithm for approximate edge connectivity is
analyzed formally in \refsection{edge-connectivity}. The basic ideas
extend to rooted vertex connectivity, with mostly superficial
differences. Compared to edge connectivity, more effort is required to
reduce global vertex connectivity to rooted vertex connectivity. The
vertex connectivity algorithms are described and analyzed in
\refsection{vertex-connectivity}.

\section{Edge connectivity}

\labelsection{edge-connectivity}

Let $G = (V,E)$ be a directed graph with edge capacities or weights
$\capacity: E \to \reals$, let $r \in V$ be a fixed root vertex, and
let $\eps \in (0,1)$ be a fixed error parameter.  Our direct aim in
this section is to compute a $\epsmore$-minimum weight $r$-cut with
high probability. A global $\epsmore$-minimum weight cut can then be
computed as either the minimum $r$-cut in $G$ or the minimum $r$-cut
in the graph reversing all the edges in $G$.

This section is exclusively concerned with edge connectivity. Cuts
will always mean edge cuts, connectivity with always mean edge
connectivity, and so forth.  For a set of vertices $S$, we let
$\incut{S}$ denote the cut of arcs going into $S$. For a set of edges
$F$, we abbreviate their total capacity by
$\capacity{F}= \sum_{e \in F} \capacity{e}$. For example,
$\incutsize{S}$ is the total capacity of the directed cut with sink
component $S$.

The high-level approach reduces approximate rooted connectivity to
logarithmically many ``well-conditioned'' ``rooted Steiner
connectivity'' problems.  The \defterm{rooted Steiner connectivity
  problem} takes as input a directed graph $G = (V,E)$, a root vertex
$r \in V$, a set of terminals $T \subseteq V - r$, and a connectivity
parameter $\lambda$. The problem is to decide, for all $t \in T$, if
the $(r,v)$-connectivity is at least $\lambda$ or not. The former can
be certified by an $(r,v)$-flow (or pre-flow) where $v$ receives
$\lambda$ units of flow; the latter can be certified by an $(r,v)$-cut
of capacity less than $\lambda$.

\subsection{Uncrossing rooted cuts}

Consider an instance of rooted Steiner connectivity, consisting of a
directed graph $G$, a set of $k$ terminals $T \subseteq V - r$, and a
connectivity parameter $\lambda > 0$. We define the
\defterm{$(\lambda,T)$-flow problem} as the flow problem with source
$r$ and an extending $G$ with an auxiliary sink vertex $\supersink$,
and arcs $(r,t)$ with capacity $\lambda$ for every terminal $t \in
T$. Intuitively, the $(\lambda,T)$-flow problem tries to
simultaneously route $\lambda$ units of flow from $r$ to every
terminal $t \in T$. The maximum $(r,\supersink)$-flow and minimum
$(r,\supersink)$-cut in this network have the following nice
properties.

\begin{lemma}
  \labellemma{wrap}
  Let $(A,B+\supersink)$ be the minimum $(r,\supersink)$-cut in the
  $(\lambda,T)$-flow problem.
  \begin{mathproperties}
  \item For all $t \in T \cap A$, the $(r,t)$-connectivity is at least
    $\lambda$.
  \item \label{wrap-property} For all $t \in T \cap B$, there is a
    minimum $(r,t)$-cut where the sink component is contained in $B$.
  \end{mathproperties}
\end{lemma}

\begin{proof}
  Let $f$ be a maximum $(r,\supersink)$-flow.  For $t \in T \cap A$,
  since $(t, \supersink)$ is saturated, $f$ carries at least $\lambda$
  units of flow from $r$ to $t$. This certifies that the
  $(r,t)$-connectivity is at least $\lambda$.

  Before addressing \cref{wrap-property}, we claim that $(A,B)$ is a
  minimum $(r,T \cap B)$-cut. To prove this, we will extract from $f$
  a maximum $(r, T \cap B)$-flow with the same capacity as $(A,B)$.

  Recall that a \defterm{preflow} is a relaxation of a flow where
  conservation of flow is relaxed to allow non-terminal vertices to
  have positive excess flow. If we restrict $f$ to $G$, then we obtain
  a $(r,T \cap B)$-preflow $f_{G}$ where each previously saturated
  terminal $t \in T \cap A$ has a surplus of $\lambda$ units of flow.
  Let $h$ be the $(r,T \cap B)$-flow on $G$ obtained from this preflow
  $f_G$ by pushing back $\lambda$ units of flow from each
  $t \in T \cap A$ back into $r$.

  \begin{center}
    \includegraphics{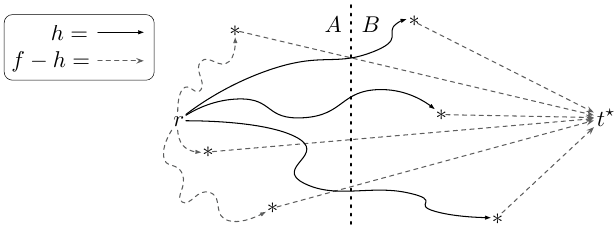}
  \end{center}

  Since $f$ is a maximum $(r,\supersink)$-flow and $(A,B+\supersink)$
  is a minimum $(r,\supersink)$-cut, $f$ saturates the $(A,B)$-cut and
  carries no flow along the reverse $(B,A)$-cut. Consequently, when we
  push back flow from some $t \in T \cap A$ back to $r$, none of the
  flow being removed goes through the $(A,B)$-cut. Thus $h$ also
  saturates the $(A,B)$-cut. As a sub-flow of $f$, $h$ also does not
  carry any flow along the reverse $(B,A)$-cut.

  Thus $h$ is an $(r, T \cap B)$-flow of size equal to the capacity of
  the $(A,B)$-cut. By duality of flows and cuts, $h$ and $(A,B)$
  mutually certify each other to be the maximum $(r,T \cap B)$-flow
  and minimum $(r,T \cap B)$-cut, respectively.

  Returning to \cref{wrap-property}, let $t \in T \cap B$. Let $U$ be
  the sink-component of a minimum $(r,t)$-cut.  By submodularity of
  cuts, we have
  \begin{align*}
    \incutsize{B} + \incutsize{U} \geq \incutsize{B \cup U} +
    \incutsize{B \cap U}.
  \end{align*}
  Since $(A \setminus U, B \cup U)$ is an $(r,T \cap B)$-cut, and
  $(A,B)$ is a minimum $(r, T \cap B)$-cut,
  \begin{align*}
    \incutsize{B} \leq \incutsize{B \cup U}.
  \end{align*}
  Consequently
  \begin{align*}
    \incutsize{U} \geq \incutsize{B \cap U},
  \end{align*}
  and $B \cap U$ is the desired subset of $B$ inducing a minimum
  $(r,t)$-cut.
\end{proof}

\subsection{Rooted preconditioning}

Given a minimum $(r,\supersink)$-cut $(A,B + \supersink)$ for a
$(\lambda,T)$-flow problem, \reflemma{wrap} allows us to contract $A$
and recurse on $G / A$ to resolve the rooted connectivities to the
remaining terminals in $T \cap B$. However we do not necessarily make
progress because we are not able to say anything about the size of $A$
or $B$. In the extreme case, the cut is $(r, V - r + \supersink)$, and
the algorithm would get stuck.

This scenario motivates the second technique of preconditioning the
graph. To this end, for a set of vertices $U$, we first define the
\defterm{in-volume of $U$}, denoted $\involume{U}$, as the sum of
unweighted in-degrees over $U$. Observe that $\involume{U}$ is a tight
upper bound on the size of the subgraph obtained by contracting
$V \setminus U$ into $r$ and dropping all arcs going into $r$.  Now,
for a parameter $\varphi > 0$, we say that $(G,r)$ is \defterm{rooted
  $\varphi$-conditioned} if $\incutsize{U} \geq \varphi \involume{U}$
for all $U \subseteq V - r$.

Given a directed graph $G$ with root $r$, and parameter $\varphi$, we
can extend $G$ and make it rooted $\varphi$-conditioned by adding, for
every non-root vertex $v \in V - r$, and arc $(r,v)$ of capacity
$\varphi \indegree{v}$, where $\indegree{v}$ is the unweighted
in-degree of $v$.

The following lemma extends \reflemma{wrap} when $G$ is
$\varphi$-conditioned to bound the size of $G/A$ proportionally to the
number of terminals in $T \cap B$.
\begin{lemma}
  \labellemma{shrink}
  Suppose $G$ is $\varphi$-conditioned, and let $(A, B+ \supersink)$
  be the minimum $(r,\supersink)$-cut in a $(\lambda,T)$-flow
  problem. Then the contracted graph $G/A$ is
  $\varphi$-conditioned and has at most
  \begin{math}
    \lambda \sizeof{T \cap B} / \varphi
  \end{math}
  edges.
\end{lemma}

\begin{proof}
  $G / A$ is rooted $\varphi$-conditioned because every rooted cut in $G / A$
  is the same rooted cut in $G$.

  As the minimum $(r,\supersink)$-cut, the capacity of the
  $(A,B + \supersink)$ is at most the capacity of the $(V,\supersink)$
  cut. The $(A, B + \supersink)$-cut has capacity
  $\lambda \sizeof{A \cap T}$ plus the capacity of the $(A,B)$-cut,
  while the $(V,\supersink)$-cut has capacity equal to
  $\lambda \sizeof{T}$. Taking their difference, the
  capacity of the $(A,B)$-cut is at most $\lambda \sizeof{T \cap B}$.

  On the other hand, because $G$ is $\varphi$-conditioned, the
  capacity of the $(A,B)$-cut is at least $\varphi \involume{B}$.
  Together, we have
  \begin{align*}
    \varphi \involume{B} \leq \lambda \sizeof{T \cap B}.
  \end{align*}
  To complete the proof, we recall that $G / A$ has at most
  $\involume{B}$ edges.
\end{proof}

\begin{lemma}
  \labellemma{shrink-wrap}
  The rooted Steiner connectivity problem with $k$ terminals and
  connectivity parameter $\lambda$ on a rooted $\varphi$-conditioned
  graph can be solved in
  $\bigO{\lambda k \log{k} / \varphi m}$ flows on graphs of
  size $\bigO{m}$.
\end{lemma}

\begin{proof}
  We prove that for $k \leq \varphi m / \lambda$ terminals, the rooted
  Steiner connectivity problem takes $\bigO{\log k}$ flows on graphs
  of size $\bigO{m}$. The more general claim follows from breaking up
  any given set of terminals into groups of size
  $\lambda m / \varphi$.

  Let $T$ be a set of $k \leq \varphi m / \lambda$ terminals. We solve
  the rooted Steiner connectivity problem for $T$ recursively.

  In the base case, $T$ consists of a single terminal $t$, and we
  compute the maximum $(r,t)$-flow and minimum $(r,t)$-cut. We either
  certify that $(r,t)$-connectivity is at least $\lambda$, or return
  the (sink component of the) minimum $r$-cut.

  In the general case, $T$ has more than one vertex. Partition $T$
  into equal halves $T_1$ and $T_2$. For $i \in \setof{1,2}$, let
  $(A_i, B_i + \supersink)$ be the minimum $(r,\supersink)$-cut in the
  $(T_i,\lambda)$-flow problem. Each terminal $t \in T_i \cap A_i$ has
  $r$-connectivity $\lambda$. We recurse on $G / A_i$ and $T_i \cap B_i$
  to address the remaining terminals. See \reffigure{shrink-wrap} for pseudocode.

  \begin{figure}[t]
    \begin{framed}
      \begin{algorithm}{shrink-wrap}{$G$,$r$,$T$,$\lambda$}
        \begin{blockcomment}
          For all $t \in T$, either certify that the
          $(r,t)$-connectivity is at least $\lambda$, or return an
          $(r,t)$-cut of capacity less than $\lambda$.
        \end{blockcomment}
      \item If $T = \setof{t}$:
        \begin{steps}
        \item Compute the minimum $(r,t)$-cut and return accordingly.
        \end{steps}
      \item Halve $T$ into $T_1$ and $T_2$.
      \item For $i = 1,2$:
        \begin{steps}
        \item $(A_i, B_i + \supersink) \gets$ minimum
          $(r,\supersink)$-cut in the $(\lambda, T_i)$-flow problem.
          \begin{blockcomment}
            All $t \in T \cap A_i$ have $r$-connectivity at least
            $\lambda$.
          \end{blockcomment}
        \item If $T_i \cap B_i$ is nonempty, then recurse on $G/A_i$ and
          $T_i \cap B_i$.
        \end{steps}
      \item Return the results gathered from above.
      \end{algorithm}
    \end{framed}
    \caption{Pseudocode of a recursive algorithm solving the rooted
      Steiner connectivity problem.}
    \labelfigure{shrink-wrap}
  \end{figure}

  Correctness of this recursive algorithm follows from
  \reflemma{wrap}.  As per the running time, each subproblem takes 2
  maximum flows over its input call excluding recursive calls.  There
  are $2^i$ subproblems of depth $i$.  Each subproblem has $k / 2^i$
  terminals and, by \reflemma{shrink},
  $\bigO{\lambda k / \varphi 2^i} \leq \bigO{m / 2^i}$ edges each. All
  flow problems at the same depth can be solved with a single flow
  taking the disjoint union of all the graphs, which forms a graph
  with $\bigO{m}$ edges. The depth of the recursion is
  $\log{\sizeof{T}}$, since the number of terminals always halves.
  Overall we solve $\log{\sizeof{T}} \leq \log{n}$ flow problems with
  $\bigO{m}$ edges.
\end{proof}

\subsection{Rooted connectivity}

To conclude this section, we now give a randomized reduction from
approximate rooted connectivity to polylogarithmically many
well-conditioned instances of rooted Steiner connectivity.

\begin{theorem}
  \labeltheorem{apx-edge-cut} A $\epsmore$-minimum rooted edge cut can
  be computed with high probability in randomized running time bounded
  by $\bigO{\log{n}^4/\eps}$ instances of single-commodity flow on
  graphs of $\bigO{m}$ edges and polynomially bounded capacities.
\end{theorem}
\begin{proof}
  At the outermost level, we binary search for a
  $\epsmore$-approximation to the true rooted connectivity
  $\optlambda$. We first note that $\optlambda$ has to be within a
  polynomial factor of some edge weight, and there are polynomially many
  powers of $1+\epsilon$ in these ranges. Thus we need $\bigO{\log n}$
  probes in a binary search to obtain a $\epsmore$-approximation to
  $\optlambda$.

  A single probe of the binary search is parameterized by a
  connectivity $\lambda$. We implement a subroutine guaranteeing that
  if the rooted connectivity is at most $\lambda$, then it will output
  a rooted cut of weight at most $\epsmore \lambda$.  Then by the
  union bound all $\bigO{\log n}$ probes succeed with high
  probability, leading to the desired $\epsmore$-approximation to
  $\optlambda$ with high probability.

  The subroutine for a given connectivity $\lambda$ is broken down
  into $\bigO{\log n}$ problems parameterized by an in-volume
  parameter $\mu$. $\mu$ will be set to each power of $2$ between $1$
  and $m$, and intuitively tries to guess the in-volume of the minimum
  $r$-cut.

  Fix $\lambda$ and $\mu$.  Let $H_{\lambda,\mu}$ precondition $G$ to
  be rooted $\varphi$-conditioned for
  $\varphi = \bigOmega{\eps \lambda / \mu}$ by adding arcs $(r,v)$ of
  weight $\eps \lambda \indegree{v} / 2 \mu$ for all $v \in V - r$. We
  also truncate the edge weights to be at most $2 \lambda$ and at
  least $\eps \lambda / 2 m$, so that the capacities in
  $H_{\lambda,\mu}$ are polynomially bounded (modulo rescaling).  Now,
  let $T \subseteq V$ sample each vertex $v \in V$ independently with
  probability $\log{n} \indegree{v} / \mu$, and run \algo{shrink-wrap}
  on the rooted Steiner connectivity problem in $H_{\lambda,\mu}$ with
  terminals $T$ and connectivity $\epsmore \lambda$. If
  \algo{shrink-wrap} returns a rooted cut of capacity less than
  $\epsmore \lambda$ in $H_{\lambda,\mu}$, then this cut also has
  capacity less than $\epsmore \lambda$ in $G$.

  Suppose there exists a set $U \subseteq V - r$ inducing an $r$-cut
  of size at most $\lambda$ and with in-volume in the range
  $\mu / 2 \leq \involume{U} < \mu$, in $G$. In $G_{\lambda,\mu}$, $U$
  induces an $r$-cut of size less than $\epsmore \lambda$. With high
  probability, $T$ contains at least one vertex from $U$. In this case
  \algo{shrink-wrap} will return a cut of capacity less than
  $\epsmore \lambda$. Over all $\mu$, if the rooted connectivity is
  less than $\lambda$, then we find a rooted cut of capacity less than
  $\epsmore \lambda$ with high probability, as desired.

  As per the running time, for fixed $\lambda$ and $\mu$, observe that
  each sample $T$ has $k = \log{n} m / \mu$ vertices in
  expectation. By \reflemma{shrink-wrap}, each call to
  \algo{shrink-wrap} takes
  $\bigO{\lambda k \log{n} / \varphi m} = \bigO{\log{n}^2 / \eps}$
  flows on graphs of size $\bigO{m}$. Over $\bigO{\log{n}}$ choices of
  $\mu$, we have $\bigO{\log{n}^3 / \eps}$ total flows on graphs of
  size $\bigO{m}$ for fixed $\lambda$. Finally, we have
  $\bigO{\log n}$ choices of $\lambda$ in the binary search, hence
  $\bigO{\log{n}^4 / \eps}$ flows overall.
\end{proof}

\begin{remark}
  When the edge weights are bounded between $1$ and $W$, the binary
  search for a $\epsmore$-approximation to $\lambda$ takes only
  $\bigO{\log{\log{nW}/\epsilon}} = \bigO{\log \log{n W} +
    \log{1/\epsilon}}$ probes, as $\optlambda$ is between $1$ and
  $m W$, and there are $\bigO{\log{nW} / \epsilon}$ powers of
  $1+\epsilon$ in this range. The logarithmic factors in the overall
  running time improve accordingly. The same holds for the vertex
  connectivity algorithms below.
\end{remark}

\section{Vertex connectivity}

\labelsection{vertex-connectivity}

This final section is about vertex connectivity, proving the bounds
for the $\epsmore$-approximation algorithms in \cref{theorem:apx-vc}
and the exact algorithm in \cref{small-vc}. At a high-level, the two
main structural ideas driving the approximate edge connectivity
algorithm, uncrossing cuts via submodularity (to ``wrap'' the
$(r,t)$-cuts of uncertified terminals $t$) and preconditioning the
graph (to ``shrink'' the contracted subgraph) are not overly specific
to edge cuts. One can develop analogous techniques in the
vertex-capacitated language, paralleling the presentation in
\refsection{edge-connectivity}.

For ease of exposition, we instead design our approximate vertex
connectivity algorithm based on a reduction to the \algo{shrink-wrap}
algorithm for edge capacities already analyzed in
\refsection{edge-connectivity}. The edge-capacitated instances are
based on the standard edge-capacitated auxiliary graph commonly used to
model vertex-capacitated graphs.

Throughout this section, let $G = (V,E)$ be a directed graph with
vertex capacities $\capacity: V \to \preals$.  For a set of vertices
$S$, we abbreviate their total capacity by
$\capacity{S} = \sum_{v \in S} \capacity{v}$.  For a set of vertices
$S$, we let $\inneighbors{S} \subseteq V \setminus S$ and
$\outneighbors{S} \subseteq V \setminus S$ denote the in-neighborhood
and out-neighborhood of $S$, respectively. For a vertex $v$, we
abbreviate $\vincut{v} = \vincut{\setof{v}}$ and
$\voutcut{v} = \voutcut{\setof{v}}$.

When in the context of rooted vertex connectivity, let $r \in V$ be a
fixed root vertex.  Let $V' = V \setminus \parof{r + N^+(r)}$. Every
$r$-cut is of the form $\vincut{S}$ for some $S \subseteq V'$.

\subsection{Rooted vertex connectivity}
We first present the $\epsmore$-approximation algorithm for the
minimum rooted vertex cut.

\begin{theorem}
  \labeltheorem{apx-rooted-vc} A $\epsmore$-minimum rooted vertex cut
  can be computed in randomized time bounded by
  $\bigO{\log{n}^4/\eps}$ single-commodity flows on graphs of
  $\bigO{m}$ edges.
\end{theorem}

\begin{proof}
  The high-level approach is very similar to that of the proof of
  \reftheorem{apx-edge-cut}.  We give a randomized reduction to
  $\bigO{\log{n}^3/\eps}$ well-conditioned instances of
  edge-capacitated rooted Steiner connectivity, in the standard
  edge-capacitated auxiliary graph modeling the vertex-capacitated
  graph $G$.

  In the exact same fashion as in the proof of
  \reftheorem{apx-edge-cut}, binary searching for the optimal
  connectivity $\optlambda$ and enumerating an in-volume parameter
  $\mu$ in powers of $2$, it suffices to implement the following
  guarantee. Given a connectivity parameter $\lambda$ and an in-volume
  parameter $\mu$, if there is a vertex $r$-cut $\inneighbors{U}$ of
  capacity less than $\lambda$ induced by a sink component $U$ of
  in-volume between $\mu/2$ and $\mu$, output an $r$-cut of capacity
  at most $\epsmore \lambda$ with high probability.

  Fix $\lambda$ and $\mu$.  Let $H$ be the standard edge-capacitated
  graph modeling the vertex capacitated graph in $G$. Namely, every
  vertex $v$ is replaced split into an arc $(\vin,\vout)$ with
  capacity equal to the capacity of $v$, and each arc $(u,v)$ is
  replaced by an arc $(\uout,\vin)$ of infinite capacity. $H$ has
  $m + n$ edges and $2n$ vertices.  It is easy to see that for all
  $t \in V'$, any finite $(\rout, \tin)$-cut in $H$ can be taken to be
  of the form
  $\incut{\Xout \cup \Uin \cup \Uout} = \setof{(\xin,\xout) \where x
    \in X}$, where $X = \incut{U}$ is a vertex $(r,t)$-cut induced by
  a set $U \subseteq V'$ containing $t$. Note that this edge cut has
  the same capacity as the corresponding vertex cut.

  Next we precondition $H$. Let $H_{\lambda,\mu}$ be obtained from $H$
  by adding, for every vertex $v \in V_{H} - \rout$, an arc
  $(\rout,v)$ of capacity $\eps \lambda \indegree{v}_H / 6 \mu$. We
  also truncate the edges capacities to be at most $2 \lambda$ and at
  least $\eps \lambda / 2 n$, making them polynomially bounded.
  $H_{\lambda,\mu}$ is $\varphi$-conditioned for
  $\varphi = \bigOmega{\eps \lambda / \mu}$.

  Let $T \subseteq V'$ independently sample each vertex $v$ with
  probability $\bigO{\log{n} \indegree{v} / \mu}$, and let
  $\Tin = \setof{\tin \where t \in T}$. We apply \algo{shrink-wrap} to
  the rooted Steiner connectivity instance in $H_{\lambda,\mu}$ with
  connectivity $\epsmore \lambda$ and terminals $\Tin$.  If
  \algo{shrink-wrap} returns an $\rout$-cut $\incut{Z}_{H'}$ of
  capacity $\leq \epsmore \lambda$, then letting
  $U = \setof{u \in V' \where \uin \in Z}$, we return the vertex
  $r$-cut in $G$ with sink component $U$. Since $\incut{Z}_{H'}$ has
  finite capacity, $\incut{Z}_{H'}$ must include the arc
  $(\xin,\xout)$ for all $x \in \inneighbors{U}$. It follows that
  $\inneighbors{U}$ is a vertex $r$-cut with capacity at most
  $\epsmore \lambda$.

  Suppose there exists a set $U \subseteq V'$ inducing a vertex
  $r$-cut $X = \vincut{U}$ of capacity less than $\lambda$, and with
  in-volume $\mu / 2 \leq \involume{U} \leq \mu$.  With high
  probability, $T$ contains at least one vertex from $U$.  In $H'$,
  the corresponding edge cut $\incut{\Xout \cup \Uin \cup \Uout}_{H'}$
  has capacity
  \begin{align*}
    \capacity{X} + \frac{\eps \lambda}{2\mu}\parof{\sum_{x \in X}
    \indegree{\xout}_H + \sum_{u \in U} \indegree{\uout}_H +
    \indegree{\uin}_H}
    +
    \frac{\eps \lambda}{4 n}  \sizeof{\incut{\Xout \cup \Uin \cup
    \Uout}_{H'}}.
  \end{align*}
  Here the second term accounts for the auxiliary arcs added to make
  $H+{\lambda,\mu}$ rooted $\varphi$-conditioned, and the third term
  accounts for truncating the edges weights to be at least
  $\eps \lambda / 4n$. To upper bound this capacity,
  we have
  \begin{math}
    \capacity{X} \leq \lambda
  \end{math}
  by choice of $X$,
  \begin{align*}
    \frac{\eps \lambda}{2\mu}\parof{\sum_{x \in X}
    \indegree{\xout}_H + \sum_{u \in U} \indegree{\uout}_H +
    \indegree{\uin}_H}
    &=                          %
      \frac{\eps \lambda}{6 \mu}\parof{\sizeof{X} +
      \sizeof{U} + \involume{U}} %
    \\
    &\tago{<}                   %
      \frac{\eps \lambda}{6 \mu} \parof{3 \involume{U}} \leq
      \frac{\eps}{2} \lambda,
  \end{align*}
  and
  \begin{align*}
    \frac{\eps \lambda}{4 n}  \sizeof{\incut{\Xout \cup \Uin \cup
    \Uout}_{H'}}
    \leq
    \frac{\eps \lambda}{4 n} \parof{2 \sizeof{X} + 2 \sizeof{U}}
    <                        %
    \frac{\eps \lambda}{2},
  \end{align*}
  for a total capacity less than $\epsmore \lambda$.  In \tagr, we
  observe that $\sizeof{U} \leq \involume{U}$ because every $u \in U$
  has in-degree at least one, and $\sizeof{X} \leq \involume{U}$
  because every $x \in X$ has an arc to some $u \in U$.  Thus $U$
  induces an $(\rout,\tin)$-cut of capacity less than
  $\epsmore \lambda$ in $H'$. It follows that \algo{shrink-wrap}
  returns an $r$-cut of weight less than $\epsmore \lambda$, as
  desired.

  The running time analysis is essentially the same as in the proof of
  \reftheorem{apx-edge-cut}.  For fixed $\lambda$ and $\mu$, we have a
  rooted Steiner connectivity problem with
  $k = \bigO{m \log{n} / \mu}$ vertices in expectation, in a rooted
  $\varphi$-conditioned graph of $\bigO{m}$ edges for
  $\varphi = \bigO{\eps \lambda / \mu}$. By \reflemma{shrink-wrap},
  each subproblem takes $\bigO{\log{n}^2 / \eps}$ flows on graphs of
  size $\bigO{m}$. We have $\bigO{\log n}$ choices of $\mu$ multiplied
  by $\bigO{\log n}$ choices of $\lambda$, giving
  $\bigO{\log{n}^4 / \eps}$ total flow problems.
\end{proof}

\subsection{Global vertex connectivity}

We now shift from rooted vertex connectivity to global vertex
connectivity. The reduction from global to rooted vertex connectivity
is not as straightforward because any particular vertex $r$ selected
as the root may actually be in the minimum global vertex cut.  That
said, there are well-known randomized reductions from global vertex
connectivity to (multiple instances of) rooted vertex connectivity
\cite{GHR00,JMY25}.  These reductions generate some nontrivial
overhead in the running time which only become exposed here when we
seek almost-linear running times. Fortunately, the overhead can be
mitigated to a logarithmic factor here in part because we only seek a
$\epsmore$-approximation. We prove the following.

\begin{theorem}
  A $\epsmore$-minimum global vertex cut can be computed with high
  probability in randomized time bounded by $\bigO{\log{n}^5/\eps}$
  single-commodity flows on graphs of $\bigO{m}$ edges and
  polynomially bounded capacities.
\end{theorem}

\begin{proof}
  We warm up with a simple but suboptimal reduction. Let
  $\ell = \bigO{\log{n} / \eps}$, and let $r_1,\dots,r_{\ell}$
  independently sample $\ell$ root vertices in proportion to their
  capacities. For each $i$, compute the minimum $r_i$-vertex cut in
  both $G$ and the graph $G_{\textsc{rev}}$ reversing all the edges in
  $G$. Return the minimum vertex cut over all these instances.

  If any $r_i$ is not in the minimum vertex cut, then we obtain a
  $\epsmore$-minimum cut with high probability. Letting $\optlambda$
  denote the vertex connectivity, all $\ell$ sampled roots are in the
  minimum vertex cut with probability
  $\parof{\optlambda / \capacity{V}}^{\ell}$. Either
  $\capacity{V} \leq \epsmore \lambda$, and any vertex cut is a
  $\epsmore$-approximate minimum cut, or else the probability of failing
  to get a good root is at most
  \begin{align*}
    \prac{\optlambda}{\capacity{V}}^{\ell} \leq \prac{1}{1+\eps}^{\ell}
    \leq \frac{1}{\poly{n}}.
  \end{align*}
  Thus we return an $\epsmore$-minimum vertex cut with high
  probability. The overall running time is a factor
  $\ell = \bigO{\log{n} / \eps}$ over the bounds for rooted connectivity
  in \reftheorem{apx-rooted-vc}.

  To remove a $(1/\eps)$-factor from the running time, we adapt an
  argument from \cite{JMY25}.  Let
  $\delta = \min{\capacity{\vincut{v}},\capacity{\voutcut{v}}}$ be the
  minimum capacity over all singleton cuts, and reduce the number of
  sampled roots $\ell$ to
  \begin{align*}
    \ell = \min{                                             %
    \bigO{\frac{\log n}{\eps}}, %
    \bigO{\frac{\capacity{V} \log n}{\capacity{V} - \delta}} %
    }.
  \end{align*}
  We have already seen that $\bigO{\log{n} / \eps}$ sampled roots
  suffice to obtain an $\epsmore$-minimum vertex cut with high
  probability. If $\ell = \bigO{c(V) \log{n} / (c(V) - \delta)}$, then
  as observed in \cite{JMY25}, the probability that all sampled roots
  are in the minimum vertex cut is at most
  \begin{align*}
    \prac{\optlambda}{\capacity{V}}^{\ell} \leq
    \prac{\delta}{\capacity{V}}^{\ell} = \parof{1 -
    \frac{c(V)-\delta}{c(V)}}^{\ell} \leq \frac{1}{\poly{n}}.
  \end{align*}
  So with high probability, at least one of the sampled roots is not in
  the minimum vertex cut, and we return an $\epsmore$-minimum vertex cut
  with high probability.

  Observe that for a particular rooted vertex connectivity problem
  with root $r$, we can remove all incoming arcs to any out-neighbors
  $u$ of $r$ except for the arc $(r,v)$ coming from $r$, without
  effecting the $r$-connectivity.  An arc $(u,v)$ is removed if either
  $v$ or an in-neighbor of $v$ other than $u$ is sampled. This occurs
  with probability
  \begin{align*}
    p_e = \frac{\capacity{v} + \capacity{\inneighbors{v}} - \capacity{u}}{\capacity{V}}.
  \end{align*}
  By linearity of expectation over all edges $e$, the expected number of
  edges we drop is
  \begin{align*}
    \sum_{e \in E} p_e
    &= \frac{1}{\capacity{V}} \sum_{(u,v) \in E}
      \capacity{v} + \capacity{\inneighbors{v}} - \capacity{u} %
    \\
    &=
      \frac{1}{\capacity{V}} \sum_{v \in V} \capacity{v} \outdegree{v} +
      \parof{\indegree{v} - 1} \capacity{\inneighbors{v}}
    \\
    &\tago{=}
      \frac{1}{\capacity{V}} \sum_{v \in V} \indegree{v}
      \capacity{\inneighbors{v}}
      \geq
      \frac{m \delta}{\capacity{V}}.
  \end{align*}
  In \tagr, the $\outdegree{v} \capacity{v}$ is distributed to each
  out-neighbor of $v$. Each vertex $v$ collects a total capacity of
  $\inneighbors{v}$, canceling out the $- \capacity{\inneighbors{v}}$.
  Thus, we expect
  \begin{align*}
    m - \sum_e p_e \leq m - \frac{m \delta}{\capacity{V}} = m \prac{c(V) - \delta}{c(V)}
  \end{align*}
  edges in the pruned subgraph for a randomly sampled root $r$. Over all
  $\ell$ roots, we expect
  \begin{align*}
    \ell m \prac{c(V) - \delta}{c(V)} \leq \bigO{m \log n}
  \end{align*}
  edges total over all pruned subgraphs.

  For each of the $\ell$ roots, we spend linear time pruning the
  graph, and then time bounded by $\bigO{\log{n}^4 / \eps}$
  single-commodity flow problems over the remaining subgraph. The
  total number of edges over all subproblems is $\bigO{m \log n}$ in
  expectation. For the sake of bounding the number of flow problems of
  size $\bigO{m}$, we can group together problems across roots so that
  the total number of edges in each flow problem is
  $\bigTheta{m}$. Altogether this gives $\bigO{\log{n}^5/\eps}$
  problems over graphs of size $\bigO{m}$. These flow problems
  dominate the $\bigO{\ell m} = \bigO{m \log{n} / \eps}$ time
  preprocessing the graph for each root.
\end{proof}

\subsection{Small vertex connecitivities}

We close with a brief discussion on small connectivities and
\cref{small-vc}.  Suppose now that $G$ has integer capacities and
global vertex connectivity $\optlambda$. Then \reftheorem{apx-vc}
implies an exact algorithm by taking
$\eps = \bigOmega{1 / \optlambda}$ running in
$\bigO{\optlambda \log{n}^5}$ flows on graphs of size $m$.

To improve the logarithmic factors, recall that the rooted
connectivity algorithm in \reftheorem{apx-rooted-vc} has an outer loop
executing a binary search for $\optlambda$ out of polynomially many
possible values, generating a $\log{n}$-factor in the running time.  A
probe for a fixed connectivity parameter $\lambda$ takes
$\bigO{\log{n}^3 / \eps}$ flows.

We make two adjustments. First, we adjust the binary search to never
exceed $2 \optlambda$, by first doubling from $1$ until reaching some
$\lambda \geq \optlambda$, and then searching within
$[1,\lambda]$. Second, for a single problem with connectivity
parameter $\lambda$, we set $\eps$ to be $1 / (1 + \lambda)$. The
first adjustment effectively replaces a $\log{n}$ factor with a
$\log{\lambda}$ factor, and the second allows us to infer the right
parameter $\eps$ without knowing $\optlambda$ \emph{a priori}.

Altogether, we have an exact algorithm for global vertex connectivity
running in time bounded by $\bigO{\lambda \log \lambda \log^4 n}$
flows on graphs of size $\bigO{m}$. For rooted connectivity we drop
the additional $\log{n}$ factor incurred by the reduction from global
vertex connectivity, giving a running time bounded
$\bigO{\lambda \log \lambda \log^3 n}$ flows on graphs of size
$\bigO{m}$, and completing the proof of \cref{small-vc}.

\printbibliography[nottype=proceedings]


\end{document}
